\documentclass{article}
\usepackage[titletoc]{appendix}

\usepackage[utf8]{inputenc}
\usepackage{geometry}
\usepackage[usenames,dvipsnames]{xcolor}
\usepackage{tikz}

\usepackage{amssymb,amsfonts,amsthm,amsmath,thmtools}
\usepackage{enumerate}
\usepackage{mathrsfs,mathtools}
\usepackage{titlesec}
\usepackage{booktabs, tabularx} 
\usepackage{graphicx} 
\usepackage{bbm}
\usepackage[hidelinks]{hyperref}
\hypersetup{
  colorlinks,
  citecolor=Violet,
  linkcolor=Black,
  urlcolor=Blue}

\usepackage{cleveref}
\usepackage{wrapfig}
\usepackage{changepage}
\def\M{{M}}
\def\K{{K}}
\def\tcK{{\tilde K}}
\def \cA {\mathcal{A}}
\def \cB {\mathcal{B}}

\def \cC {{C}}

\def \cD {\mathcal{D}}
\def \cE {\mathcal{E}}
\def \cF {\mathcal{F}}
\def \cG {\mathcal{G}}
\def \cH {\mathcal{H}}
\def \cI {\mathcal{I}}
\def \cJ {\mathcal{J}}
\def \cK {\mathcal{K}}
\def \cL {\mathcal{L}}
\def \cM {\mathcal{M}}
\def \cN {\mathcal{N}}

\def \cQ {\mathcal{Q}}

\def \bE {\mathbb{E}}
\def \bF {\mathbb{F}}

\def \bN {\mathbb{N}}

\def \bP {\mathbb{P}}
\def \P {\mathbb{P}}
\def \bQ {\mathbb{Q}}
\def \R {\mathbb{R}}

\def \bT {\mathbb{T}}

\def\1{{\mathbf 1}}

\def \sC {\mathscr{C}}

\def\cone{\operatorname{cone}}
\def\conv{\operatorname{conv}}

\DeclareMathOperator*{\esssup}{ess\,sup}

\let\phi\varphi
\let\epsilon\varepsilon
\def \eps {\epsilon}

\usepackage{xspace}

\newcommand{\one}[1]{\mathbbm{1}_{#1}}
\newcommand{\ind}[1]{\mathbbm{1}_{\left\{{#1} \right\}}}

\newcommand{\condE}[2]{\mathbb{E}\left[ \left. #1 \, \right| #2 \right]}
\newcommand{\E}{\mathbb{E}}
\newcommand{\V}{\mathbf{V}}
\def\tV{{\tilde\V}}
\def\tv{{\tilde v}}

\newcommand{\wt}[1]{ \widetilde{ #1 }}
\newcommand{\wh}[1]{ \widehat{ #1 }}

\newcommand{\for}{ \text{ for }}
\newcommand{\q}[1]{\quad \text{ #1 } \quad}

\def\qiq{\quad \implies \quad}

\newcommand{\ul}[1]{{#1}}
\newcommand{\ol}[1]{\overline{#1}}

\renewcommand \emptyset \varnothing

\def\cprod{{\sC}}

\newtheorem{thm}{Theorem}[section]
\newtheorem{theorem}[thm]{Theorem}
\newtheorem{cor}[thm]{Corollary}
\newtheorem{prop}[thm]{Proposition}
\newtheorem{lem}[thm]{Lemma}
\newtheorem{lemma}[thm]{Lemma}

\newtheorem{remark}[thm]{Remark}
\theoremstyle{definition}
\newtheorem{defn}[thm]{Definition}

\newtheorem{example}[thm]{Example}

\theoremstyle{remark}

\renewcommand\thmcontinues[1]{Continued}

\title{ Multi-currency reserving for coherent risk measures}
\usepackage{authblk}

\author[1,2]{Saul Jacka\thanks{Saul D. Jacka gratefully acknowledges funding received from the EPSRC grant EP/P00377X/1 and is also grateful to the Alan Turing Institute for their financial support under the EPSRC 
grant EP/N510129/1. E-mail: \textit{s.d.jacka@warwick.ac.uk}\\Department of Statistics\\University of Warwick\\Coventry CV4 7AL, UK}}
\author[1]{Seb Armstrong \thanks{Seb Armstrong gratefully acknowledges funding received from the EPSRC Doctoral Training Partnerships grant EP/M508184/1. E-mail: \textit{s.armstrong@warwick.ac.uk}\\Department of Statistics\\University of Warwick\\Coventry CV4 7AL, UK}}
\author[3]{Abdelkarem Berkaoui \thanks{E-mail: \textit{berkaoui@yahoo.fr}\\College of Sciences\\Al-Imam Mohammed Ibn Saud Islamic University\\P.O. Box 84880\\Riyadh 11681\\Saudi Arabia\\}}

\affil[1]{University of Warwick}
\affil[2]{Alan Turing Institute}
\affil[3]{Imam Muhammad ibn Saud Islamic University}

\begin{document}

\maketitle

\begin{abstract}
We examine the problem of dynamic reserving for risk in multiple currencies under a general coherent risk measure.  The reserver requires to hedge risk in a time-consistent manner by trading in baskets of currencies. We show that reserving portfolios in multiple currencies $\V$ are time-consistent when (and only when) a generalisation of  Delbaen's m-stability condition \cite{D06}, termed optional $\V$-m-stability, holds. We prove a version of the Fundamental Theorem of Asset Pricing in this context. We show that this problem is equivalent to dynamic trading across baskets of currencies (rather than just pairwise trades) in a market with proportional transaction costs and with a frictionless final period.
\end{abstract}

\section{Introduction}

Coherent risk measures (CRMs) were introduced in \cite{ADEH99}. A key example was  based on  the Chicago Mercantile Exchange's margin requirements. The Basel III accords mandate the use of Average Value at Risk (a coherent risk measure unlike the widely-used Value at Risk (VaR)measure, which is not coherent) for reserving risk-capital for certain derivatives-based liabilities \cite{BCBS}. Many financial institutions have regulatory or other reasons for testing their reserves and  a dynamic version of coherent risk measures is a model for this process. 

In \cite{JBA} we outlined an approach to reserving for risk based on CRM's. The potential drawback of reserving with CRM's, as has been pointed out repeatedly, is the problem of time-consistency (see, for example \cite{Biel} and references therein): one can view the time-$t$ reserve for a liability payable at a later time $T$ as itself a liability, payable at time $t$. A serial version of this shows that (for example) a regulator who imposes the reserving requirements implicit in the CRM  is actually requiring a sequence of reserves $\rho_t(X)$ -- one at each time-point where reserves are audited-- for a liability $X$, and consequently it can be argued that one actually needs an initial reserve of $\rho_0 \circ \dots \circ \rho_{T-1}(X)$. Delbaen \cite{D06} gave a necessary and sufficient condition, termed multiplicative stability (henceforth m-stability), for this latter quantity to equal $\rho_0(X)$, which does not hold in general, although the inequality
\[ \rho_0 \circ \dots \circ \rho_{T-1}(X) \ge \rho_0(X)\] 
does. In particular, Average- or Tail-Value at Risk (also known as Expected Shortfall) is not, in general, time-consistent.

It is normally assumed, in the context of CRM's, that assets and liabilities are discounted to time-0 values. Since CRM's are measures of {\em monetary} risk for amounts payable at time $t$, we think it is clearer to take the {\em prospective view} that liabilities are expressed in terms of time-$T$ units and so at time 0, the risk or reserve is expressed in terms of units of a zero-coupon bond  (or currency) payable at $T$. Of course, as soon as one adopts this approach it is clear that our assets need not just correspond to the unit of account and we should consider the possibility of holding multiple currencies or assets to perform the reserving function. In \cite{JBA} we showed how multiple currencies allowed the possibility of an extended version of time-consistency: predictable $\V$-time consistency. We envisaged a set of assets numbered $0,1,\dots, d$ with random terminal values $\V = (v^0, v^1,\dots, v^d)$ (given in the distinguished unit of account) and gave a necessary and sufficient condition (Theorem 2.15 of \cite{JBA}) for time-consistent, multi-asset reserving to work for any specific CRM.

Examples of CRMs include superhedging prices in incomplete frictionless markets and (as we shall see) minimal hedging endowments in markets with proportional transaction costs.

In this paper we consider a stronger version of multi-asset time-consistency which corresponds to explicitly adjusting portfolios (and which therefore seems appropriate to situations where trading of the assets held as reserves is possible) which includes both these situations. We term this version \emph{optional time-consistency}, and see it as the appropriate setting for many situations, including those mentioned above.

We shall give  necessary and sufficient conditions for $\cA$, the cone of acceptable claims corresponding to a CRM, $\rho$, to be expressible as the (closure in the appropriate topology of the)  sum, over times $t$, of trades in the underlying assets which  are acceptable at time $t$ (\Cref{thm:main}). We will then show that under this condition we obtain a version of the Fundamental Theorem of Asset Pricing for CRMs (\Cref{thm:FTAP}). Finally, in \Cref{thm:tradingCone} we shall show the equivalence between optionally time-consistent CRMs and a generalisation (corresponding to permitting trades in baskets of assets) of the models for trading with proportional transaction costs introduced by Jouini and Kallal \cite{JK}, developed by Cvitanic and Karatzas \cite{cvitanic1996hedging},  Kabanov \cite{kabanov1999hedging}, Kabanov and Stricker (see \cite{kabanov2001harrison}) and further studied by Schachermayer \cite{schachermayer2004fundamental} and Jacka, Berkaoui and Warren \cite{JBW},  amongst others. For more recent developments see  Bielecki, Cialenco and Rodriguez \cite{Biel2} and their survey paper  \cite{Biel}.

\section{Preliminaries}

Insurers reserve for future financial risks  by investing in suitably prudent and sufficiently liquid assets, typically bonds, or any other asset universally agreed always to  hold positive value. We call such assets num\'eraires, examples of which include paper assets, such as currencies, and physical commodities.  Reserving a sufficient amount ensures that the risk carried by the insurer is acceptable to the insurer and (possibly) to regulatory authorities, customers and their agents. In some circumstances, the choice of num\'eraire is  clear; in others, it is not,  for example when insurers reserve for claims in multiple currencies. It is common to calculate reserves by a \lq\lq prudent'' calculation of expected value in a pessimistic or \lq\lq worst realistic case'' scenario. We assume that the minimal amount sufficient to form the reserve is modelled by a coherent risk measure (CRM); see \cite{FS11} for an introduction to CRM's. 

We assume the availability of a finite 
collection of num\'eraires numbered $(0,\ldots,d)$.
We examine the problem of reserving for a risk at a terminal time $T$, through adjusting the reserving portfolio held in  the num\'eraire \lq\lq currencies'' at discrete times $t=0,1,\dots,T$. The terminal value of the num\'eraires (in units of account) is denoted $\V = (v^0,v^1,\dots,v^d)$, and we assume that each $v^i$ is a strictly positive, $\cF_T$-measurable, bounded random variable, with the Euclidean norm of $\V$ bounded away from 0. 
Thus, to value any portfolio $Y$ of holdings in the elements in $\V$, we take the inner product $Y\cdot \V$ and, conversely, any bounded $X$ may be written in the form $Y\cdot \V$ with $Y$ bounded . We regard the portfolio $Y$ as corresponding to a {\em liability} of   $Y\cdot \V$ at terminal time $T$.

A coherent risk measure is a reserving mechanism: we  assume that an insurer is  reserving for risk according to a conditional coherent risk measure $\rho_t$, at each time $t$. They reserve  the amount $\rho_t (X)$ for a random claim $X$. Thus the aggregate position of holding the risky claim $X$ and reserving adequately should always be acceptable to the insurer. The set $\cA_t$ of acceptable claims at time $t$ consists of those $\cF_T$-measurable bounded random variables with non-positive $\rho_t$. We shall say that the portfolio $Y_t$ reserves at time $t$ for a claim $X$ if
\[ X - Y_t \cdot \V \in \cA_t.\]

\subsection{Notation}

We fix a terminal time $T\in \bN$, a discrete time set $\bT := \{0,1,\dots, T\} $. We fix a probability space $(\Omega,\cF, \P)$, where $\bP$ is the \emph{reference measure} or \emph{objective measure}. The filtration $(\cF_t)_{t\in \bT}$ describes the information available at each time point. The space of all  $\cF$-measurable random variables is denoted $L^0 = L^0(\Omega,\cF,\P)$; we denote $ L^0(\Omega,\cF_t,\P)$ by $L^0_t$. The space of all $\P$-integrable (respectively $\P$-essentially bounded) $\cF_t$-measurable random variables is $L^1_t$ (resp. $L^\infty_t$). The space of $\cF_t$-measurable (respectively integrable; essentially bounded) $\R^{d+1}$-valued random variables is denoted $\cL^0_t= L^0(\Omega,\cF_t,\P;\R^{d+1})$ (resp. $\cL^1_t$; $\cL^\infty_t$). A subscript `+' denotes the positive orthant of a space, and `++' denotes strict positivity; for example, the set of non-negative (respectively strictly positive) essentially bounded random variables is $L^\infty_+$ (resp. $L^\infty_{++}$). Similarly, a subscript \lq -' denotes the negative orthant of a space. 

\subsection{Definitions and well-known results}
We recall some definitions and concepts.
At each time $t \in \bT$, we wish to price monetary risks using all information available at that time. The following definition is adapted from \cite{DS05}:
\begin{defn}\label{defn:ccrmp}
A map $\rho_t : L^\infty \to L^\infty_t$ for $t \in \bT$ is a \emph{conditional convex risk measure} if, for all $X, Y \in L^\infty$, it has the following properties:
\begin{itemize}
\item
Conditional cash invariance: for all $m \in L^\infty_t$,
\[ \rho_t (X+ m ) = \rho_t(X) + m\qquad  \P\text{-almost surely;}\]
\item
Monotonicity: if $X \le Y$ ${\P}$-almost surely, then $\rho_t(X) \le \rho_t(Y)$;
\item
Conditional convexity: for all $\lambda \in L^\infty_t$ with $0 \le \lambda \le 1$,
\[ \rho_t(\lambda X + (1- \lambda)Y) \le \lambda \rho_t(X) + (1-\lambda)\rho_t(Y) \qquad  \P\text{-almost surely;}\]
\item
Normalisation: $\rho_t(0)=0$ ${\P}$-almost surely.
\end{itemize}
Furthermore, a conditional convex risk measure is called \emph{coherent} if it also satisfies 
\begin{itemize}
\item
Conditional positive homogeneity: for all $\lambda\in L^\infty_t$ with $\lambda \ge 0$,
\[ \rho_t(\lambda X)= \lambda\rho _t (X) \qquad  \P\text{-almost surely.} \]
\end{itemize}
\end{defn}
Our interest lies chiefly in reserving for and pricing liabilities. We see a positive random variable $X$ as a liability to be hedged, and a negative $X$ as a credit, which explains our (non-traditional) choice of sign in the cash invariance property, and the direction of monotonicity.
\begin{defn}
A convex risk measure satisfies the \emph{Fatou property} if, for any bounded sequence $(X^n)_{n \ge 1} \subset L^\infty$ converging to $X\in L^\infty$ in probability, we have 
\[ \rho_t(X) \le \liminf_{n\to \infty} \rho_t(X^n).\]
\end{defn}
  The Fatou property is equivalent to \emph{continuity from above}: $\rho_t$ is continuous from above if, whenever $(X^n)_{n \ge 1} \subset L^\infty$ is a non-increasing sequence such that $X^n \to X$ $\P$-a.s., then
\[ \rho_t(X^n) \to \rho_t(X) \qquad \P\text{-a.s. as }n \to \infty \]
As shown in \cite{ADEH99}, whenever the Fatou property holds, we may represent $\rho_t$ as
\[ \rho_t(X)=\esssup_{ \bQ \in \cQ} \bE_{ \bQ}[X| \cF_t] .\]
The representing set of probability measures $\cQ$ is dominated by $\P$, that is, that $\bQ \ll \P$ for any $\bQ \in \cQ$. In this and the subsequent sections we identify the probability measures $\bQ$ of the set $\cQ$ with their Radon-Nikodym derivative, or density, with respect to $\bP$, $\frac{d\bQ}{d\P}$. We trust that which version is being used will be clear from the context. We will also denote the density of $\bQ$ by $\Lambda^{\bQ}$ and the density of the restriction of $\bQ$ to $\cF_t$ by $\Lambda^{\bQ}_t$. Recall that the resulting process is a $\bP$-martingale.

The \emph{time-$t$ acceptance set} of a conditional convex pricing measure $\rho_t : L^\infty \to L^\infty_t$ is 
\[ \cA_t = \{ X \in L^\infty : \rho_t(X) \le 0\}. \]

For the following results, we refer the reader to \cite{FS11} and \cite{DS05}. We equip the space $ L^\infty$ with the weak$^*$ topology $\sigma( L^\infty,  L^1)$, so  that the topological dual will be $ L^1$. Recall that a set $\cC$ of claims is arbitrage-free whenever
\[ \cC \cap L^\infty_+ \subseteq\{ 0\},\]
or more precisely $\cC \cap L^\infty_+ \subseteq\{ [0]\},$ where $[0]$ denotes the equivalence class of random variables which are almost surely equal to 0.
\begin{prop}\label{prop:background}
Define $(\cA_t)_{0\leq t\leq T}$ to be the acceptance sets of the dynamic conditional \emph{coherent} pricing measure  $\rho_t : L^\infty \to L^\infty_t$ satisfying the Fatou property and set $\cA_0=\cA$.

Then  $ \cA_t$ is a weak$^*$-closed\footnote{in $ L^\infty$, i.e., $\cA_t$ is closed in the topology $\sigma( L^\infty_t,  L^1_t)$} convex cone that is  stable under multiplication by bounded positive $ \cF_t$-measurable random variables, contains  $ L^\infty_- $, and is arbitrage-free.
\end{prop}
\begin{remark} In general, dynamic conditional coherent risk measures are defined on a stream of payments $(X_t)_{0\leq t\leq T}$. However the standard trick of extending the reference measure $\P$ to a probability measure on $\Omega\times \{0,\ldots T\}$ enables us to reduce our study to the single payment case (see section 3.3 of \cite{JBA} for more details of how this is done).
\end{remark}

\section{Optional representation and multi-currency time consistency}

\subsection{Time-consistency}

An insurer who has insured the claim $X$ needs to hold a sequence of portfolios $Y_0,Y_1,\dots, Y_T$ (one for each time point at which a reserve calculation is to be made) so that the risk is adequately reserved for, and so that no unacceptable risk is assumed in any one exchange of portfolios. That is to say, in the optional case, from time $t-1$ until just before time $t$, the insurer holds a portfolio $Y_{t-1} \in \cL^\infty(\cF_{t-1})$  of the num\'eraires as an acceptable reserve for $X$, and will wish to exchange to a new reserving portfolio $Y_{t}$. The insurer may only exchange to the new portfolio $Y_t$ if the risk of the adjustment is acceptable, i.e. $\rho_t(Y_t\cdot \V - Y_{t-1}\cdot \V) \le 0$. Thus all the transfer of risk occurs instantaneously at time $t$ (we shall see in section \ref{proptran} that the analogy with a trading set-up is no coincidence).  This is  in contrast to the predictable case developed in \cite{JBA}, where the idea is that the time-$(t-1)$ reserve is an adequate reserve for the hedging portfolio needed at time $t$. In the predictable case, the acceptable risk is carried between the time points $t-1$ and $t$, whereas in the optional case an explicit exchange of known amounts of the num\'eraires needs to take place at time $t$ to update the reserve portfolio. 

We shall say that the dynamic risk measure is (optionally) $\V$-time-consistent if this property holds (at least in a limiting sense) for each claim $X$, starting from an initial reserve $\rho_0(X)$.

\begin{defn} A dynamic convex risk measure $\rho=(\rho_t)_{t=0,\dots,T}$ is \emph{optionally $\mathbf{V}$-time-consistent} if,  for any $X\in \cA$, we may find a sequence $X^n$ in $\cA$ and a sequence $\pi^n =(\pi^n_t)_{t=0,\dots,T-1}$ such that $\pi^n_t \in \cL^\infty(\cF_t)$ for each $t$, and
\begin{enumerate}[(i)]
\item
\begin{equation}\label{c1} X^n\rightarrow X \text{ almost surely}
\end{equation}
\item
for each $t$,
\[\rho_t(\pi^n_t \cdot\mathbf{V})\le 0 \; \bP \text{ a.s.}; \]
\item
for each $n \in \bN$,
\[\sum_{t=0}^{T-1} \pi^n_t \cdot\mathbf{V} = X^n \quad  \P\text{-almost-surely. }\]
\end{enumerate}
\end{defn}
\begin{remark}
By the subsequence property, we can replace the almost sure convergence in (\ref{c1}) by convergence in $L^0$ without affecting the definition.
We think of $\sum_0^t \pi^n_t$ as the reserve held at time $t$ for the liability $X^n$.
\end{remark}

\subsection{Representability of claims}
\label{sec:representability}

We may view optional $\V$-time-consistency as a condition on the sequences of portfolios that can superhedge a claim $X$. Given a $\V$-time-consistent dynamic CRM  $(\rho_t)$, we may (at least in a limiting sense) express $X$ as the sum of the initial reserve $\rho_0(X)$ and the $(T+1)$ adjustments at times $0,1,\dots T$ (where we set $Y_{-1}$ to be any vector in $\cL_0$ with $\rho_0(Y.\V)=\rho_0(X)$, 
for example $\frac{\rho_0(X)}{\rho_0(\V .1)}\1$):

\[ 
X = \rho_0(X)  + \sum_{t=0}^{T} (Y_{t} - Y_{t-1}) \cdot\mathbf{V}, 
\]
where each adjustment satisfies $\rho_t((Y_{t} - Y_{t-1}) \cdot\mathbf{V}) \le 0$. Each adjustment $Y_{t} - Y_{t-1}$ is an $\cF_t$-measurable portfolio with $t$-acceptable valuation; we call the set of such portfolios $K_t(\cA,\V)$. We seek to answer the question \lq\lq Is it possible to represent every claim in $\cA$ by a series of such adjustments?'' 

Given any cone $\cD $ in $\ul L^\infty$ and our vector $\mathbf{V}$ of num\'eraires, we define the collection of portfolios attaining $\cD $ to be
\[ \cD (\mathbf{V}) = \{ Y \in \cL^\infty: Y\cdot \mathbf{V}\in \cD  \}.\]
The set of time-$t$ acceptable portfolios that are $\cF_{t}$-measurable is denoted
\begin{equation}\label{ktdef}
 K_t(\ul{\cA}, \mathbf{V}) := \ul{\cA}_t(\mathbf{V}) \cap \ul{\cL}^\infty(\cF_t).
\end{equation}
\begin{defn}\label{defn:predrep}The cone $\ul{\cA}$ in $L^\infty$ is said to be \emph{optionally $\V$-representable} if
\begin{equation}\label{repdef} \ul{\cA} (\mathbf{V}) = \ol{\oplus_{t=0}^{T} K_t(\ul{\cA}, \mathbf{V})},
\end{equation}
where the closure is taken in the weak$^*$ topology. If this is the case, we also say that $\ul{\cA}$ is \emph{optionally represented} by $\mathbf{V}$.
When $\mathbf{V}$ is fixed, we also say that $ \ul{\cA} (\mathbf{V}) $ is optionally represented if (\ref{repdef}) holds.
\end{defn}
\begin{remark}
It is an easy exercise to show that 
\begin{equation}\label{ex}K_t(\ul\cA,\mathbf{V}) = \{ X \in\ul \cL^\infty(\cF_{t},\R^{d+1}): \alpha X \in \cA \text{ for any } \alpha \in \ul L^\infty_+(\cF_t) \}.
\end{equation}
This characterisation is used repeatedly in what follows and in the proof of Theorem \ref{thm:crucialClaim}.
\end{remark}
From now on, where there is no ambiguity, we shall write $\K_t$ for $K_t(\ul\cA,\mathbf{V})$

\subsection{Stability}
We recall Delbaen's m-stability condition, on a standard stochastic basis $(\Omega,\cF,(\cF_t)_{t=0,\ldots,T},\P)$:
\begin{defn}[Delbaen \cite{D06}] A set of probability measures $\cQ \subset L^1(\Omega,\cF,\P)$ is \emph{m-stable} if for elements $\bQ^1, \bQ^2 \in \cQ$, with associated density martingales $\Lambda^{\bQ^1}_t = \condE{\frac{d\bQ^1}{d\P}}{\cF_t}$ and $\Lambda^{\bQ^2}_t = \condE{\frac{d\bQ^2}{d\P}}{\cF_t}$, and for each stopping time $\tau$, the martingale $L$ defined as
\[ L_t = \begin{cases}\Lambda^{\bQ^1}_t &\for t\le \tau \\  \frac{\Lambda^{\bQ^1}_\tau}{\Lambda^{\bQ^2}_\tau}\Lambda^{\bQ^2}_t &\for t \ge \tau \end{cases}\]
defines an element, $\bQ$, in $\cQ$. The probability measure $\bQ$ is also defined by the properties that
\[
\bQ\big|_{\cF_\tau}=\bQ^1\big|_{\cF_\tau}\text{ and }\bQ(\cdot|\cF_{\tau})=\bQ^2(\cdot|\cF_{\tau})\; \bP\text{a.s.},
\]
\end{defn}
\noindent so $\bQ$ pastes together the laws $\bQ^1$ and $\bQ^2$ at time $\tau$.

We generalise  m-stability by allowing extra freedom over one time period when pasting two measures together {\em and} by only pasting measures satisfying a consistency condition relating to $\mathbf{V}$:

\begin{defn}\label{def:optionalPastings} Let $\tau$ be a stopping time, and $\bQ^1,\bQ^2$ be two probability measures absolutely continuous with respect to $\bP$. The set $\bQ^1 \oplus^{\text{opt}}_\tau \bQ^2$ of \emph{optional pastings of $\bQ^1$ and $\bQ^2$} consists of all $\wt \bQ$ such that 
\begin{enumerate}[(i)]
\item $ \tilde\bQ \big|_{\cF_\tau}=\bQ^1 \big|_{\cF_\tau}$, 
\item[]and
\item for any $A \in \cF_T$, $\wt \bQ(A |\cF_{(\tau+1)\wedge T})=\bQ^2(A |\cF_{(\tau+1)\wedge T})$.
\end{enumerate}
\end{defn}
We make explicit the freedom over the time period $(\tau, \tau+1]$ by writing any optional pasting in terms of the two measures being pasted, and a ``one-step density'':
\begin{lem}
For $\tau$  a stopping time, and $\bQ^1,$ $\bQ^2$ two probability measures, 
\[ \bQ^1 \oplus^{\text{opt}}_\tau \bQ^2 = \left\{ \wt \bQ \ll \P : \Lambda^{\wt\bQ} = \Lambda^{\bQ^1}_\tau R \frac{\Lambda^{\bQ^2}}{\Lambda^{\bQ^2}_{(\tau+1)\wedge T}}, \quad \text{for some }R \in L^1_+(\cF_{(\tau+1)\wedge T}) \text{ s.t. } \condE{R}{\cF_\tau}=1  \right\}\]
\end{lem}

\begin{defn}\label{prop:equivstab}
The set of probability measures $\cQ$ is \emph{optionally $\V$-m-stable} if,   whenever $\tau$ is a stopping time,  $\bQ^1, \bQ^2 \in \cQ$, and $\wt \bQ \in \bQ^1 \oplus^{\text{opt}}_\tau \bQ^2$ has the (additional) property that
\begin{equation}\label{eq:pmstabcond} 
\bE_{\wt \bQ} [\V|\cF_\tau] = \bE_{ \bQ^1} [\V|\cF_\tau],
\end{equation}
then $\wt \bQ$ is also in $\cQ$.
\end{defn}
\begin{example}
It is easy to check that given an $(\cF_t)_{0\leq t \leq T}$-adapted and bounded process $X$, the collection, $\cQ_X$, of Equivalent Martingale Measures for $X$ is optionally $X_T$-m stable.  
\end{example}

\noindent Note that a set that is optionally 1-m-stable is automatically m-stable but the converse is false.

The following proposition gives an equivalent definition of optional $\V$-m-stability in terms of the dual cone $\cA(\V)^*$.

\begin{prop}\label{prop:eqstab} Suppose, without loss of generality that  the set of pricing measures, $\cQ$, is  convex and closed (so the set of densities is closed in the topology of $L^1$), and let $\cD  = \cA(\V)^*$. The following are equivalent:
\begin{enumerate}[(i)]
\item
$\cQ$ is optionally $\V$-m-stable
\item
for each $t\in \{0,1,\dots, T\}$, whenever $Y,W \in \cD $ are such that there exists $Z\in \cD $, an event $F\in \ul{\cF}_t$, positive random variables $\alpha, \beta \in \cL^0_+(\ul{\cF}_{t+1})$ with $\alpha Y, \beta W \in \ul{\cL}^1$ and $X:= \one{F} \alpha Y + \one{F^c} \beta W$ satisfies
\begin{equation} \label{eq:conestabcond}
\condE{X}{\ul{\cF}_t}=\condE{Z}{\ul{\cF}_t},
\end{equation}
 then $X$  is a member of $\cD $.
\end{enumerate}
\end{prop}

The proof can be found in \Cref{sec:lempf}. 

\begin{defn}We shall say that an arbitrary  cone $\tilde\cD\subseteq \cL^1_T$ satisfying condition (ii) of \Cref{prop:eqstab} is \emph{optionally m-stable}.
\end{defn}

\begin{lem}\label{dualv} Suppose that $\mathbf{V}$ is a collection of $d+1$ num\'eraires, and $\cD $ is a convex cone in $L^\infty$. Then \[ \cD (\mathbf{V})^*=\cD ^*\mathbf{V}.\]
 \end{lem}
The proof can be found in \Cref{sec:lempf}.

\subsection{An Equivalence Theorem}

Our first result is a set of conditions equivalent  to optional $\V$-time-consistency, including a precise statement of $\V$-representability, and a dual characterisation which pertains to the convex set of probability measures $\cQ$ that define the risk measure. 

This result resembles that obtained in \cite{JBA} for  \emph{predictable} versions of these concepts. 

To show the equivalence of $\V$-m-stability and $\V$-representability, we find  the dual of each $\K_t$, which we  call the \emph{optional pre-image} of $\cA(\mathbf{V})^*$ at time $t$. Aside from its  utility in proving the equivalence of $\mathbf{V}$-optional representability and optional $\mathbf{V}$-m-stability, the optional pre-image of an optionally m-stable convex cone $\cA(\mathbf{V})^*$ at time $t$ is a concrete description of the dual of the set of portfolios  held at time $t$ in order to maintain an acceptable position at time $t$.

We fix   the vector of num\'eraires $\mathbf{V}$,  a coherent pricing measure  $\rho=(\rho_t)_t$ with a closed, convex representing set of probability measures $\cQ$, and take $\ul{\cA}_t$ to be the acceptance set of $\rho_t$ for $t \in \bT$. One of our two main results is

\begin{theorem}\label{thm:main}
The following are equivalent:
\begin{enumerate}[(i)]
\item
$(\ul\rho_t)_{t\in \bT}$ is optionally $\mathbf{V}$-time-consistent;
\item
$\ul{\cA}$ is optionally represented by  $\mathbf{V}$;
\item
 $\cQ $ is optionally $\V$-m-stable.
\end{enumerate}
\end{theorem}

\begin{example}\label{bigex}[A generic example]
Given  a positive $X\in \cL^1_T$, and a sequence of random, closed, convex sets $\cI :=(I_t)_{t=0,\ldots,T}$ in $\R^{d+1}$, each  measurable with respect to $\cE(\R^{d+1},\cF_t)$, the relevant Effros $\sigma$-algebra (see Remark 4.2 of \cite{JBW}), let
$$
\cQ^X_{\cI}:=\{\bQ\sim \bP:\: \E_{\bQ}[X|\cF_t]\in I_t\text{ for each }t\},
$$
then $\cQ^X_{\cI}$ is optionally $X$-m-stable. Note that $X$ need not be in $\cL^\infty$.

To recover the case of EMMs, simply take $X$ to be $M_T$, the terminal value of a positive $\bP$-martingale and $I_t$ to be the singleton $\{M_t\}$. Of course, $M_T$ is not necessarily in $\cL^\infty$, but we may rectify this by taking $\V=(v_0,\ldots,v_d)$, where
$v_i=\frac{M_T^i}{\sum_j M_T^j}$, and letting $\cQ$ be defined as the set  $\cQ:=\{ \bQ:\;\Lambda^{\bQ}=\frac{\sum_j M_T^j}{\sum_j \E_{\tilde\bQ} M_T^j}\Lambda^{\tilde \bQ} \text{ for some }\tilde\bQ\in \cQ^{M_T}_{\cI}\}$.
\end{example}

\label{sec:pomr}
We give the proof of \Cref{thm:main} in two steps.  First, we will show equivalence of optional $\V$-representability and optional $\V$-m-stability.
The proof of the equivalence of optional $\V$-time-consistency and optional representability is given after we have proved \Cref{thm:FTAP} -- a version of the Fundamental Theorem of Arbitrage Pricing.

\begin{defn}
For $\cD \subset \ul{\cL}^1_+$, we define, for each time $t$, the \emph{optional pre-image} of $\cD$ by
\begin{align}\label{eq:MtD }
\cM_t(\cD ) := \{ Z \in \ul{\cL}^1:  &\exists \alpha_t \in \ul{L}^0_{t,+},  \exists Z'\in \cD  \nonumber
\\&\text{ such that } \alpha_t Z' \in \ul{\cL}^1(\R^{d+1}) \text{ and } \condE{Z}{\ul{\cF}_{t}}= \alpha_t \condE{Z'}{\ul{\cF}_{t}} \}.
\end{align}
\end{defn}

The optional pre-image of  a set $\cD \subset \ul{\cL}^1_+$ is key in understanding optionally stable convex cones, as shown in the following two lemmas:

\begin{lem}\label{lem:A1} Suppose $\cD  \subset \ul{\cL}^1_+$. If $\cD $ is an optionally stable convex cone, then
\[ \cD = \bigcap_{t=0}^{T} \cM_t(\cD ).\]
\end{lem}

If $S\subset \cL^1$, we denote by  the  $ \ol{\conv}(S)$  the closure in $\cL^1$ of the convex hull of $S$.
\begin{lem}\label{lem:A2} Suppose $\cD  \subset \ul{\cL}^1_+$, and define
\[ [\cD ] := \bigcap_{t=0}^{T}\left( \ol{\conv} \cM_t(\cD )\right), \] 
(where $ \cM_t(\cD )$ is as defined in \eqref{eq:MtD }).
Then
\begin{enumerate}[(a)]
\item
$[\cD ]$ is the smallest stable closed convex cone in $\ul{\cL}^1$ containing $\cD $;
\item
$\cD =[\cD ]$ if and only if $\cD $ is a stable closed convex cone in $\ul{\cL}^1$.
\end{enumerate}
\end{lem}
We prove both these lemmas in \Cref{sec:lempf}.  The proof of equivalence of statements (ii) and (iii) of \Cref{thm:main} is dependent on  the following
\begin{theorem}\label{thm:crucialClaim}
For any $t\in \{0,1,\dots,T-1\}$, 
\begin{equation}\label{eq:MandK}
\K_t=(\cM_t(\ul\cA(\mathbf{V})^*))^*.
\end{equation}
\end{theorem}
The proof is given in \Cref{sec:lempf}.

Thus we characterise each \lq\lq summand'' in the representation (cf. \Cref{defn:predrep}) as the dual of the optional pre-image of the dual of the set of acceptable portfolios in $\mathbf{V}$.

\begin{proof}[Proof of \Cref{thm:main}, equivalence of (ii) and (iii)] 
By assumption, $\cA(\V)$ is a weak$^*$-closed convex cone in $\ul\cL^\infty(\R^{d+1})$ which is arbitrage-free, so that $\cA(\V)^{**} = \cA(\V)$. 
Recall that $\cA(\V)$ is optionally representable if 
\[\cA(\V) = \ol{\oplus_{t=0}^{T} \K_t}.\]
 Thanks to Proposition \ref{prop:eqstab}, we must show the equivalence of the two conditions
\begin{enumerate}[(i')]
\setcounter{enumi}{1}
\item
$\cA(\V)$ is optionally representable; and
\item
$\cA(\V)^*$ is optionally $\V-m$-stable.
\end{enumerate}

\emph{(ii') $\Rightarrow$ (iii'):} Assuming $\cA(\V)$ is optionally representable, it follows from \Cref{thm:crucialClaim} that
\[ \cA(\V) = \ol{\oplus_t \K_t} =  \ol{\oplus_t \cM_t(\cA(\V)^*)^*}.\]
Taking the dual, we find that
\[   \cA(\V)^* = \cap_t \cM_t(\cA(\V)^*)^{**} = \cap_t \ol{\conv}\cM_t(\cA(\V)^*)  \]
where the second equality follows from the Bipolar Theorem. Hence, $ \cA(\V)^* = [\cA(\V)^*]$, and so by \Cref{lem:A2}, $\cA(\V)^*$ is optionally stable.

\emph{(iii') $\Rightarrow$ (ii'):} Assuming $\cA(\V)$ is a weak$^*$-closed convex  cone, note that $\cA(\V)^*$ is a convex cone closed in $(\ul\cL^1, \sigma(\ul\cL^1, \ul\cL^\infty))$. Assuming further that $\cA(\V)^*$ is \emph{stable},
\begin{align*}
\cA(\V)^* &= \cap_t \cM_t(\cA(\V)^*) &\text{by \Cref{lem:A1}}
\\&=  \cap_t \K_t^* &\text{by \cref{eq:MandK}}.
\end{align*}
Now we may apply the Bipolar Theorem to deduce
\[ \cA(\V)\equiv \cA(\V)^{**} = \ol{\oplus_{t} \K_t}\]
and $\cA(\V)$ is optionally representable, as required.

\end{proof}

\section{The Fundamental Theorem of Multi-currency Reserving  }\label {fund}

As announced in the introduction, we now discuss closure properties in $\cL^0$ of the decomposition of a $\V$-optionally representable acceptance set $\cA$.

By analogy to the definition in \cite{schachermayer2004fundamental}, we define a trading cone as follows:
\begin{defn}
$\cC\subset \cL^0(\R^{d+1},\cF_t)$ is said to be a {\em (time-$t$) trading cone} if
$\cC$ is closed in $\cL^0$ and is closed under multiplication by non-negative, bounded, $\cF_t$-measurable random variables.
\end{defn}
We recall Lemma 4.6 of \cite{JBW} which we quote here (suitably rephrased) for ease of reference:
\begin{theorem}
\label{abscone}  Let $\cC$
be a closed convex cone in $\cL^0(\cF)$, then
\begin{equation}\label{stable}
\cC\hbox{ is stable under multiplication by (scalar) elements
of }L^\infty_+(\cF)
\end{equation}
if and only if there is a random closed cone $M^{\cC}$ such that
\begin{equation}\label{stable2}
\cC= \{X\in \cL^0(\cF):\; X\in M^{\cC}\text{ a.s.}\}.
\end{equation}
\end{theorem}

We shall demonstrate that if $\cA$ is $\V$-representable then $(\K^0_t)_{0\leq t\leq T}$, the   $L^0$-closures of the cones $\K_t$, are trading cones, whose sum is closed, and equal to the $L^0$-closure of $\cA(\V)$, which is is arbitrage-free. 

This is a version of the (First) Fundamental Theorem of Asset Pricing (FTAP).

For the rest of this section closures in $L^0$ or $\cL^0$ will be denoted by a simple overline, whereas weak${}^*$ closure of a set $S$ will be denoted $\ol{S}^w$
We set $\cA^0(\V) := \ol{\cA(\V)}$, the closure of $\cA(\V)$ in $\cL^0$. Recall that $\cA^0(\V)$ is arbitrage-free whenever
\[ \cA^0(\V) \cap \cL^0_{+}=\{0\},\]
and define the trading cone 
\[ \cC_t =\{X \in \cL^0_t: cX \in \cA^0(\V) \text{ for all }c \in L^\infty_+(\cF_t)\}. \]
Note that closure in $\cL^0$ of $\cC_t$ follows immediately from the closure of $\cA^0(\V)$.
\begin{lemma}\label{k}
For each $t$, $\K^0_t$ is a trading cone and if $X\in \cL^0_t$ then $X\in K^0_t$ iff $X.\E_{\bQ}[V|\cF_t]\leq 0 \text{ for all }\bQ\in\cQ.$
\end{lemma}
\begin{proof}
Now if $X\in \cL^\infty_t$ then $X\in \K_t$ if and only if $\E_{\bQ}[X.V|\cF_t]=X.\E_{\bQ}[V|\cF_t]\leq 0 \text{ for all }\bQ\in\cQ.$ It follows that
 $\K^0_t=\{X\in  \cL^0_t:\; X.\E_{\bQ}[V|\cF_t]\leq 0 \text{ for all }\bQ\in\cQ\}$ and this is obviously a trading cone.
\end{proof}

It follows from Theorem \ref{abscone} and Lemma \ref{k} that
\begin{lemma}
There are random closed cones $M^\cC_t$ and $M^K_t$ such that 
\[ \cC_t =  \{Y \in \cL^0_t: Y \in M^\cC_t \text{ a.s.}\}\]
\[ K_t^0 =  \{Y \in \cL^0_t: Y \in M^K_t \text{ a.s.} \}\]
and the polar (in $\R^{d+1}$) of $M^K_t$ is
$\ol{cone}(\{\E_{\bQ}[\V|\cF_t]:\; \bQ\in\cQ\})$, the random closed cone generated by $\{\E_{\bQ}[\V|\cF_t]:\; \bQ\in\cQ\}$.
\end{lemma}

We now give the main theorem of this section:
\begin{theorem}\label{thm:FTAP}
The set $\cG:=\oplus_t \cC_t$ is closed in $L^0$, arbitrage-free and equals $\cH:=\oplus_t K^0_t(\cA,\V) $. Moreover, if $\cA$ is $\V$-representable, then their common value is $\cA^0(\V)$ and then $\cA^0$, the closure in $L^0$ of $\cA$ is given by
\begin{equation}\label{closure2}
\cA=\cA^0(\V).\V=\oplus_t K^0_t(\cA,\V).\V.
\end{equation}
\end{theorem}

\begin{proof}
The proof is in three steps. We will show that:
\begin{itemize}
\item[1.] $\cG$ is closed in $\cL^0$. 
\item[2.]$C_t=K^0_t(\cA,\V)$ (and $\cG$ is arbitrage-free) establishing equality of $\cG$ and $\cH$. 
\item[3.]$\cA^0(\V)=\cH$ if $\cA$ is $\V$-representable and $\cA^0=\cA^0(\V).\V$.
\end{itemize}
{\em Proof of 1.} We recall Definition 2.6 and Lemma 2.7 (suitably rephrased) from \cite{JBW}
\begin{defn}\label{null}
Suppose $\cJ$ is a sum of convex cones in $\cL^0$:
$$
\cJ=\M_0+\ldots +\M_T.$$ 
We call elements of $\M_0\times\ldots\times
\M_T$ whose components almost surely sum to 0, null-strategies (with respect to
the decomposition $\M_0+\ldots +\M_T$) and denote the set of them by
$\cN(\M_0\times\ldots\times \M_T)$. 
\end{defn}
For convenience we denote
$\cC_0\times\ldots\times \cC_T$ by $\cprod$.
\begin{lemma}{\em (Lemma 2 in
Kabanov et al \cite{kabanov3})}
\label{s} Suppose that
$$\cJ=\M_0+\ldots +\M_T$$
is a decomposition of $\cJ$ into trading cones; then $\cJ$ is closed if
$\cN(\M_0\times\ldots\times \M_T)$ is a vector space and each
$\M_t$ is closed in $\cL^0$.
\end{lemma}
Since we have already established that each $\cC_t$ is a trading cone, applying Lemma \ref{s} to the decomposition of $\cG$, we only need to prove that the null strategies $\cN(\cprod)$ form a vector space. The argument is standard: since $\cG$ is a cone, we need only show that $\xi=(\xi_0,\ldots,\xi_T)\in\cN(\cprod)$ implies that $-\xi\in\cN(\cprod)$. To do this, given $\xi \in\cN(\cprod)$, fix a $t$ and a bounded non-negative $c\in L^0_t$ with a.s. bound $b$. Then, since $\xi$ is null,
$$
-c\xi_t =b\xi_0+\ldots b\xi_{t-1}+(b-c)\xi_t+b\xi_{t+1}+\ldots +b\xi_T,
$$
and each term in the sum is clearly in the relevant $\cC_s$ and hence in $\cA^0$. Since $c$ and $t$ are arbitrary, $-\xi_t\in\cC_t$ for each $t$ and so $-\xi\in\cN(\cprod)$.

It is clear from (\ref{ex}) that $\K_t\subseteq \cC_t$ and hence, by closure of $\cC_t$ that $\K^0_t\subseteq\cC_t$. Thus $\cH\subseteq \cG$.

{\em Proof of 2.} Recall from \cite{JBW} that consistent price processes for $\cH$ are those martingales valued in $(M^K_t)^*$ at each time step. Since $\Lambda^{\bQ}_t\V^{\bQ}_t$ is such a martingale (for any $\bQ\in \cQ$), the collection of consistent price processes for the sequence of trading cones $ K^0_t(\cA,\V) $ is non-empty and so, by Theorem 4.11 of \cite{JBW}, $\ol{\cH}$ is arbitrage-free.

 The consistent price processes for $\oplus_t \cC_t$ are those martingales valued in $(M^\cC_t)^*$ at each time step.
We now claim that, for each $t$, 
\[ (M^\cC_t)^* = (M^K_t)^*.\]
Once we establish this, equality follows on taking the random polar cones in $\R^{d+1}$.

First, observe that $\cC_t \supseteq K^0_t(\cA,\V) $ implies $(M^\cC_t)^* \subseteq (M^K_t)^*$ almost surely. So, assume that $(M^\cC_t)^* $ is a strict subset of $ (M^K_t)^*$. Then there exists $\bQ \in \cQ$ such that 
\[ \P(\bE_\bQ[\V|\cF_t] \not \in M^\cC_t ) >0. \]
For this $\bQ$, we form the consistent price process $Z_t = \Lambda^{\bQ}_t  \bE_{\bQ}[\V |\cF_t] \in (M^K_t)^*$. Form the frictionless trading cones 
\[ C_t(Z) := \{ X \in L^0_t: X \cdot Z_t \le 0 \}\]
and we have an arbitrage-free and closed cone $\wt \cA = \oplus_t C_t(Z)$ from the FTAP. Clearly $\wt \cA$ contains $\cA^0(\V)$, and so $C_t(Z)$ is contained in $\cC_t$, whence $Z_t \in M^\cC_t$ a.s., contradicting the assumption of strict inclusion.

{\em Proof of 3.} If $\cA$ is $\V$-representable then 
$$
\cA^0(\V)=\ol{\ol{(\oplus_t \K_t)}^w}=\ol{\oplus_t\K^0_t}=\ol{\cH}
$$
 but, as we have already established, $\cH$ is closed. Finally, since $\cA\supset K_t.\V$ it is clear that $\cA^0\supseteq \cH.\V$. Conversely, since $\cA^0=\ol{\cA(\V).\V}=\ol{\ol{\oplus K_t}^w.\V}=\ol{\oplus K_t.\V}$, it follows that $\cA^0\subseteq \cH.\V$
\end{proof}
\section{Completing the Proof of \Cref{thm:main}}

\begin{proof}[Proof of \Cref{thm:main}: the equivalence of (i) and (ii)]
We shall use the result from \Cref{thm:FTAP} that if optionally $\V$-representable then 
$$
\cA^0(\V)=\oplus_t K^0_t(\cA,\V)
$$
and
\begin{equation}\label{closure}
\cA^0=\oplus_t K^0_t(\cA,\V).\V,
\end{equation}
where the superscript 0 represents closure in $L^0$ or $\cL^0$.
Now define 

{\em condition (iv)}:
\begin{equation}\label{cond:iv}
\cA\subseteq \oplus_t K^0_t(\cA,\V).\V.
\end{equation}
Clearly (\ref{closure})$\Rightarrow$(\ref{cond:iv}) and hence (ii)$\Rightarrow$(iv).
It is  clear that
 (iv)$\Rightarrow$(i). Conversely, from \Cref{thm:FTAP}, $\oplus_tK^0_t$ is closed in $\cL^0$, whereas (i) tells us that $\cA\subset \ol{\oplus K_t.\V}=\ol{\oplus_tK^0_t}$ so we conclude that (i)$\Rightarrow$(iv). Now it is sufficient to prove that (iv)$\Rightarrow$(ii).

Suppose (iv) holds. We shall show that $\cA\subseteq \ol{\oplus_t K_t.\V}$ or, equivalently (since $\cA$ is a closed convex cone), that 
\begin{equation}\label{iv}
(\oplus_t K_t.\V)^*\subseteq \cA^*.
\end{equation}
Define
$$
G_t=(\oplus_{s=t}^T K_s.\V)^* \text{ and }B_t:=L^\infty \cap \oplus_{s=t}^T K^0_s.\V
$$
Now, since $(L^\infty\cap\bigl(\oplus_tK^0_t.\V\bigr))^*\subseteq \cA^*$ we may show (\ref{iv}) by proving, by induction that
\begin{equation}\label{induct}
G_t\subseteq B_t^*\text{ for each }t.
\end{equation}
Clearly (\ref{induct}) holds for $t=T$ since $K_T=\{X\in\cL^\infty:\; X.\V\leq 0\text{ a.s. }\}$ and $K_T^0=\{X\in\cL^0:\; X.\V \leq 0\text{ a.s.}\}$ so $B_T=L^\infty_{-}=K_T.\V$.

Now suppose that (\ref{induct}) holds for $t=u+1$. Take arbitrary $Z\in G_u$ and $X\in B_u$. Then $X=\alpha_u.\V +Y$, for some 
$Y\in  \oplus_{s=u+1}^T K^0_s.\V$ and $\alpha_u\in K^0_u$. For integer $n>0$, set $F_n=\{||\alpha_u||\leq n\}$, then 
$\alpha_u\1_{F_n}\in K_u$ and $Y\1_{F_n}=X\1_{F_n}-\alpha_u\1_{F_n}\in L^\infty$ (since $X\in L^\infty$). Since, for any $t\geq u$,  $K^0_t$ is 
closed under multiplication by $\1_{F_n}$, it follows that $Y\1_{F_n}\in B_{u+1}$ with $X\1_{F_n}=\alpha_u.\V 
\1_{F_n}+Y\1_{F_n}$ and hence $X\1_{F_n}\in B_u$. Since $Z\in G_{u+1}$ it follows from the induction hypothesis that
$\E ZX\1_{F_n}=\E Z(Y\1_{F_n}+\alpha_u.\V \1_{F_n})\leq \E Z(\alpha_u.\V \1_{F_n})$. Now $\alpha_u.\V\1_{F_n}\in K_u.\V$ and 
$Z\in (K_u.\V)^* $ so $\E[ZX\1_{F_n}]\leq 0$. Thus, by dominated convergence, $\E[ZX]\leq 0$ and since $X$ is an arbitrary element 
of $ B_u$ it follows that $Z\in B_u^*$ and so $G_u\subseteq B_u^*$, establishing the inductive step.
\end{proof}
We have actually proved something extra: 
\begin{cor}If $\cA$ is $\V$-representable then $B_0$ is weak* closed and equal to $\cA$.
\end{cor}
\begin{proof}
From \Cref{thm:FTAP}, If $\cA$ is $\V$-representable then $\cA^0=\oplus_{s=0}^TK^0_s.\V$, so
$$
B_0=L^\infty \cap  (\oplus_{s=0}^T K^0_s.\V)=L^\infty\cap \cA^0\supset \cA.
$$
Now since $\cA$ is $\V$-representable, $\cA=\ol{\oplus_t K_t.\V}\supset B_0^{**}$, from (\ref{induct}), but 
$ B_0^{**}=\ol{B_0}$ so 
$$
\cA=\ol{B_0}=B_0.
$$
\end{proof}
\begin{remark}
Of course, if $\cA$ is $\V$-m-representable then we could replace the almost sure convergence in the definition of optional time consistency by weak${}^*$-convergence, at the cost of taking a net rather than a sequence.

It follows from the corollary that if $\cA$ is $\V$-time-consistent then, given $X\in \cA$ we can form a sequence of reserves $R_t$ such that $R_T.\V=X$, and each increment $\pi_t:=R_t-R_{t-1}$ is in $K^0_t$ so that $\pi_t$ is the almost sure limit of elements of $K_t$.
\end{remark}

\section{Associating a pricing mechanism to a market with proportional transaction costs}\label{proptran}
Having made the connection in \Cref{fund} between optionally-representable CRM's and trading cones, in this section, we directly associate the reserving mechanism to a hedging strategy in a market with transaction costs. This is achieved by adding an extra time period  $(T,T+1]$ to the market with transaction costs, in which all positions are cashed out into a base num\'eraire   $v^0$. We do this by imposing num\'eraire risks that are so disadvantageous as to force a risk-averse agent to sell-up at time $T$, rather than in the additional period.

Let $e_0=(1,0,\dots,0),\dots,e_d=(0,\dots,0,1)$ denote the canonical basis of $\R^{d+1}$. 
Recall that in a market with transaction costs the basic set-up has a collection of assets (labelled $0,\dots,d$) and  random bid-ask prices $\pi^{i,j}_t$ at each trading time $t\in\{0,\dots,T\}$. Thus $\pi^{i,j}_t$ is the number of units of asset $i$ that can be exchanged for one unit of asset j at time $t$. The corresponding trading cone, which we denote by $\tcK^0_t(\pi_t)$ is generated by these trades together the possibility of consumption so that $\tcK_t(\pi_t)$ is the (closed) cone generated by non-negative  $\cF_t$-measurable multiples of the vectors $-e_i$ and $e_j - \pi^{ij}_t e_i$, for $i,j \in \{0,1,\dots,d\}$. 
The set of claims available from zero endowment is then
\[ \cB_T(\pi) = \bigoplus_{t=0}^{T} \tcK^0_t(\pi_t).\]
We (initially) assume that the closure of $\cB_T(\pi)$ in $\cL^0$ is arbitrage-free. Note that thanks to Theorem 1.2 of \cite{JBW} we may (and shall) then assume  that, by amending the bid-ask prices if necessary,
$\cB_t(\pi)$ is closed. The proof of this theorem also establishes 
that the null strategies for the resulting trading cones form a vector space.

We denote the $L^0$-closure of the set of acceptable claims under a risk measure generated by a collection of absolutely continuous probability measures,$\cQ$,   by $\cA_\cQ^0$. We will show that each market corresponds to a  CRM

\begin{theorem}\label{thm:tradingCone}
For the sequence of transaction cost matrices $(\pi_t^{ij})_{t=0,1,\dots,T}$,  , there is a stochastic basis $(\wt\Omega, \wt\cF, \wt \bF, \wt \P)$,  a vector of num\'eraires $\V\in L^\infty(\wt\Omega, \wt\cF, \wt \P; \R^{d+1})$, and a set of optionally $\V$-m-stable probability measures $\cQ$ such that the closure (in $\cL^0$) of the corresponding set of $\cF_T$-measurable attainable claims is  the collection of claims attainable by trading in the underlying assets:
\[
 \cB_T(\pi)= \cA_{\cQ}^0(\V) \cap L^0(\cF_T). \]
 \end{theorem}

The key element in the proof is to add an extra trading period $(T,T+1]$ at the end in which all positions are cashed out into asset 0. However, we impose num\'eraire risks that are so disadvantageous as to force the agent to sell up in the preceding time period, rather than in the additional period. To generate the final, frictionless prices, we add on a simple ``coin spin'' for each other asset. We encode the $d$ binary choices (either buy or sell each of the other $d$ num\'eraires) as $\{0,1\}^d$, and define $\mu$ to be the uniform measure on $(\{0,1\}^d, 2^{\{0,1\}^d}) $. Thus, we define the augmented sample space $\wt \Omega := \Omega \times \{0,1\}^d$, and define the product sigma-algebra and measure:
\begin{align*}
\wt \cF	&:= \cF \otimes 2^{\{0,1\}^d}),
\\ \wt \P &:= \P \otimes \mu.
\end{align*}
We augment the filtration trivially, by setting $\wt \cF_t := \cF_t \otimes \{ \emptyset, \{0,1\}^d\}$. We employ the obvious embedding of $L^\infty(\Omega,\cF_t,\P)$ in $L^\infty(\wt\Omega,\wt\cF_t,\wt\P)$;  it should be clear from context  to which version of $L^0$ we are referring.

Fix $0<\eps<1$ small. Define the $\R^{d+1}$-valued random variable $\tV = (\tv^0,\tv^1, \dots, \tv^d)$, for  $\omega \in \Omega$, $\omega' \in \{0,1\}^d$, by
%
\begin{align}
\tv^0(\omega,\omega') &= 1,
\\  \tv^i(\omega,\omega') &= (1-\omega'_i)(1-\epsilon)\frac{1}{\pi^{i,0}_T(\omega)} + \omega'_i (1+\epsilon)\pi^{0,i}_T(\omega). \label{eq:Vfrompi}
\end{align}

The interpretation of $\tV$ is this: arriving at time $T$ at a bid-ask spread $\left[ \frac{1}{\pi^{i,0}_T(\omega)}, \pi^{0,i}_T(\omega) \right]$  for num\'eraire $i$ in state $\omega$, we spin a coin. If the coin shows heads ($\omega'_i = 1$), the $(T+1)$-price of asset $i$ is slightly higher than the $T$-ask price, and any  negative holding of $i$ to time $T+1$ makes a loss compared to cashing out at time $T$. If the coin shows tails ($\omega'_i = 0$), the $(T+1)$-price of asset $i$ is slightly lower than the $T$-bid price, and any  positive holding of $i$ makes a  loss. Any risk-averse agent will seek to avoid these losses by cashing out into asset 0 at time $T$.

Now we define the frictionless bid-ask matrix at time $T+1$ by
\[ \pi^{ij}_{T+1} := \frac{\tv^j}{\tv^i}.\]
The trading cone $\tcK^0_{T+1}(\pi_{T+1})$ is generated by positive  $\cF_{T+1}$-measurable multiples of the vectors $-e_i$ and $e_j - \pi^{ij}_{T+1} e_i$, for $i,j \in \{0,1,\dots,d\}$. Define the cone $\cB_{T+1}(\pi) = \cB_{T}(\pi)+\tcK^0_{T+1}(\pi_{T+1})$.

The collection of consistent price processes for the original set of claims $\cB_T(\pi)$ is 
\[ \cB^\circ_T(\pi) = \{ Z \in L^1_T(\R^{d+1}): \bE[Z| \cF_t] \in \tcK^0_t(\pi_t)^* \text{ a.s. for  }t=0,1,\dots,T \}.\]
By Theorem 4.11 of \cite{JBW}, since $\cB_T(\pi)$ is closed and has no arbitrage,
 there exists at least one consistent price process $Z$ for $\cB_T(\pi)$. The following proposition shows that the cone $\cB_{T+1}(\pi)$ is arbitrage-free.
\begin{prop}
There is a consistent price process for $\cB_{T+1}(\pi)$.


\end{prop}

\begin{proof}
We extend any consistent price process for $\cB_T(\pi)$ to a consistent price process for $\cB_{T+1}(\pi)$ by multiplying by the Radon-Nikodym derivative for the martingale measure for each coin spin. 
For any $Z\in \cB^\circ_T$, define $\lambda^Z >0$  such that the one-period process $\left(Z^i/Z^0, \, \lambda^Z \tv^i\right)$ is a $\wt \P$-martingale for each $i$. Then 
\begin{equation}\label{eq:extraPeriodCPP}
Z_{T+1} = Z^0 \lambda^Z \tV 
\end{equation}
defines a consistent price process for the cone $\cB_{T+1}(\pi)$.

We first show that such a $\lambda^Z$ always exists. Note that $Z \in \tcK^0_T(\pi_T)^*$ gives that, $\omega$-a.e.,
\[
\frac{Z_T^j(\omega)}{Z_T^i(\omega)} \le \pi_T^{ij}(\omega) \le \pi_T^{i,0}(\omega) \pi_T^{0,j}(\omega) < \frac{1+\epsilon}{1-\epsilon}\pi_T^{i,0}(\omega) \pi_T^{0,j}(\omega)  = \frac{\tv^j(\omega,1)}{\tv^i(\omega,0)}, 
\]
with $\tv^j(\omega,1)$ understood to be $\left.\tv^j(\omega,\omega') \right|_{\omega'_j=1}$ etc.
Fixing $\omega \in \Omega$ and $i\neq 0$, we see that
\[ \ol Z_T^i(\omega) := \frac{Z_T^i(\omega)}{Z_T^0(\omega)} \in (\tv^i(\omega,0),\tv^i(\omega,1)).\]
The martingale measure for such an one-period binary tree model is determined by the probability of ``heads''
\[ \theta(\omega,i) = \frac{\ol Z^i_T-\tv^i(\omega,0)}{\tv^i(\omega,1)-v^i(\omega,0)}.\]
Now set 
\[ \lambda^Z(\omega,\omega') = 2^d\prod_{i=1}^d \theta(\omega,i)^{\omega'_i} (1-\theta(\omega,i))^{1-\omega'_i}\]
Clearly,  $\lambda^Z$ is a.s. positive and bounded, $\tilde\E[\lambda^Z|\cF_T]=1$ and  $Z^i_{T+1}\in \cL^1$ since $\tilde \E[Z^i_{T+1}]
=\tilde \E[\lambda^Z\tV_i Z^0_T]=\E[Z^i_T]$ by the the positivity of $Z^i_{T+1}$, Fubini's Theorem and the definition of $\mu$ and $\lambda^Z$
 
Similarly, for any $X_T \in L^\infty_T(\R^{d+1})$, 
\begin{equation} \label{eqn:extendCPPexp}
\bE_{\wt \P}[X_T \cdot Z_{T+1}] = \bE_{ \P}[X_T \cdot Z^0_T\bE_{\wt \P}[( \lambda^Z \tV)|\cF_T]] =  \bE_{ \P}[X_T \cdot Z_{T}].
\end{equation}
Setting $X_T = \one{A} e_i$ for $A \in \cF_T$, for any $i$, we see that $Z_T = \bE_{\wt \P}[Z_{T+1}|\cF_T]$, and $Z_{T+1}$ is thus a consistent price process as required.
\end{proof}

\begin{prop}
The cone $\cB_{T+1}(\pi):=\oplus_{t=0}^{T+1}\tcK^0(\pi_t)$ is closed in $L^0$, and is arbitrage-free.
\end{prop}

\begin{proof}

From Theorem 4.11 of \cite{JBW}, we have that the closure of $\cB_{T+1}(\pi)$ in $L^0$ is arbitrage-free. We will show  that the set of null strategies $\cN(\tcK^0_0(\pi_0),\dots,\tcK^0_T(\pi_T),\tcK^0_{T+1}(\pi_{T+1}))$ is a vector space, and  conclude from Lemma \ref{null} that the cone $\cB_{T+1}(\pi)$ is closed in $L^0$, and we are done.

 Take 
\[ (x_0, \dots,x_{T+1}) \in \cN(\tcK^0_0(\pi_0),\dots,\tcK^0_T(\pi_T),\tcK^0_{T+1}(\pi_{T+1})).\]
Let 
\begin{equation}\label{trick}
x = x_0+\dots+x_T\text{, so that }x+x_{T+1}=0.
\end{equation}
 We see that $x_{T+1}$ is an $\cF_T$-measurable element of $\tcK^0_{T+1}(\pi_{T+1})$. We claim that $x_{T+1} \in\cB_T(\pi)$: since $x_{T+1}\in\cB_{T+1}$, for any consistent price process $Z$ for the cone $\cB_{T+1}(\pi)$, and any $n \in \bN$,
\[ 0\geq  \bE[Z_{T+1}x_{T+1} \ind{||x_{T+1}|| < n}] =\bE[Z_Tx_{T+1} \ind{||x_{T+1}|| < n}] =,\]
so $x_{T+1} \ind{||x_{T+1}|| < n} \in \cB_T(\pi)$ for all $n$, and so $x_{T+1} \in\cB_T(\pi)$ by closure of $\cB_T(\pi)$.

Since $x_{T+1} \in\cB_T(\pi)$, there exist $y_0 \in \tcK^0_0(\pi_0), \dots, y_T\in \tcK^0_T(\pi_T)$ with $x_{T+1}= y_0+\dots+y_{T} $. 
Then, rewriting (\ref{trick}), we see that
\[
(x_0+y_0)+\ldots + (x_T+y+T)=0.
\]
Since each term in this sum is in the relevant trading cone, we see that $(x_0+y_0),\ldots ,(x_T+y_T))$ is in $\cN(\tcK^0_0,\ldots,\tcK^)_T$. Now, by assumption, this is a vector space so that each $-(x_t + y_t) \in \tcK^0_t(\pi_t)$, and so, since  $\tcK^0_t(\pi_t)$ is a cone containing $y_t$, each $-x_t$ is in $\tcK^0_t(\pi_t)$.

The bid-ask prices are frictionless at time $T+1$, so  $x_{T+1} \in \tcK^0_{T+1}(\pi_{T+1})$ may be written as $u_1-u_2$, where $u_1\in lin(\tcK^0_{T+1}(\pi_{T+1}))$, and $u_2 \ge 0$. Note that
\[ 0 \le u_2 = u_1-x_{T+1}=u_1 + x \in \cB_T(\pi),\]
but since  $\cB_T(\pi)$ is arbitrage-free, $u_2=0$, and so $-x_{T+1} \in \tcK^0_{T+1}(\pi_{T+1})$ and thus the set of null strategies is a vector space.
\end{proof}

The final prices $\tV$ above are, in general unbounded, so we transform these by normalising, setting
$$
\V=(v_0,\ldots,v_d) \text{ where }v_i :=\frac{\tv_i}{\sum_j\tv_j}.
$$

Finally, we define the set of measures
\[ 
\cQ = \{\bQ^Z:\; Z\text{ is a consistent price process for }\cB_T(\pi)\},
\text{ where }
\frac{d\bQ^Z}{d\wt\P}:= \frac{Z^0_T \lambda^{Z_T} (\tv_0)}{\sum_j Z^j_0}.
\]
It is easy to check that these are probability measures from the fact that the $Z$'s are consistent price processes and hence strictly positive, vector-valued martingales.

The proof of the main result is now clear:

\begin{proof}[Proof of \Cref{thm:tradingCone}]
We observe that 
\begin{align*}
\K^0_t(\cA_{\cQ}(\V))&=\{X\in\cL^0_t:\; X.\E_{\bQ}[\V|\cF_t]\leq 0 \text{ a.s. for all }\bQ\in\cQ\}\\
&=\{X\in\cL^0_t:\; X.Z_t\leq 0 \text{ a.s. for all consistent }Z\}\\
\end{align*}
It follows that 
\[
\cK^0_t(\pi_t)=\{Y\in\cL^0_t:\; Y.Z_t\leq 0 \text{ a.s. for all consistent }Z\}=
\{Y\in\cL^0_t:\; Y\in\K^0_t(\cQ)\}
\]
and so
\[
\cB_T=\cA^0_{\cQ}
\]
\end{proof}






\bibliographystyle{plain}


\appendix

\section{Appendix}

\subsection{Proofs of subsidiary results}\label{sec:lempf}

\begin{proof}[\textbf{Proof of \Cref{dualv}}]
First take $Z\in \cD ^*$. For any $X\in \cD (\mathbf{V})$ we have $\bE[ZV\cdot X] \le 0$ and so $ZV \in \cD (\mathbf{V})^*$, thus $\cD (\mathbf{V})^* \supseteq \cD ^*\mathbf{V}$. 

For the reverse inclusion, denote by $e_i$ the $i$th canonical basis vector in $\R^{d+1}$. First, since $\mathbf{V}\cdot \alpha(v^ie_j-v^je_i)=0$, we have
\[ \alpha(v^ie_j-v^je_i) \in \cD (\mathbf{V}) \qquad \forall \alpha \in L^\infty.\]
Take $Z \in \cD (\mathbf{V})^*$. Now, for any $i,j \in \{1,\dots,d\}$, $\alpha \in L^\infty$, we have
\[ \bE[Z \cdot  \alpha(v^ie_j-v^je_i)] \le 0.\]
Reversing $i$ and $j$ in the above, we may write $\bE[Z \cdot  \alpha(v^ie_j-v^je_i)] = 0$, and allowing first $\alpha= \ind{Z \cdot  (v^ie_j-v^je_i) >0}$ then $\alpha= \ind{Z \cdot  (v^ie_j-v^je_i) <0}$, we see that in fact,
\[ Z\cdot  (v^ie_j-v^je_i) = 0 \qquad \text{a.s. for any }i,j,\]
and so, taking $i=0$ we have $ Z^j=Z^0v^j $ a.s. for each $j$, thus any $Z \in \cD (\mathbf{V})^*$ must be of the form $Z^0\mathbf{V}$ for some $Z^0 \in L^1$. Now, given $C\in \cD $, take $X$ such that $X\cdot \mathbf{V}=C$ (which implies that $X\in \cD (\mathbf{V})$), then 
\[0 \ge \bE[W\mathbf{V}\cdot X]=\bE[WC],\]
and since $C$ is arbitrary, it follows that $W\in \cD ^*$. Hence $\cD (\mathbf{V})^*\subseteq \cD ^*\mathbf{V}$.
\end{proof}

\begin{proof}[\textbf{Proof of \Cref{prop:eqstab}}] By \Cref{dualv}, we may write $\cD  = \cone \{ \frac{d\bQ}{d\P} \V : \bQ \in \cQ\}$.

\emph{(i)$\implies$(ii):} We suppose that (i) holds, and fix $t \in \{0,1,\dots, T\}$, $Y,W,Z \in \cD$, $F \in \cF_t$, $\alpha,\beta \in \cL^0_+(\cF_{t+1})$ and $X \in \cL^1_+$ as in the hypothesis of (ii). We show that $X \in \cD$ by applying (i) twice.
First, take $\tau = t \one{F} +T\one{F^c}$, and define $\bQ^Z$ and $\Lambda^Z_t$ via 
\[\frac{d\bQ^Z}{d\P} = \frac{Z^0}{\bE[Z^0]}  \q{and}  \Lambda^Z_t = \condE{\frac{d\bQ^Z}{d\P}}{\cF_t}  \]
and analogously for $\bQ^Y$, $\Lambda^Y$. Note that $Z = \bE[Z^0] \Lambda^Z \V$. We now form an optional pasting of $\bQ^Z$ and $\bQ^Y$ at time $\tau$, as $\wh \bQ$, via
\[ \wh \Lambda = \Lambda^Z_t \left(\frac{\alpha \bE[Y^0]\Lambda^Y_{t+1} }{\bE[Z^0]\Lambda^Z_t} \right) \frac{\Lambda^Y}{\Lambda^Y_{t+1}} \one{F} + \Lambda^Z \one{F^c}. \]
This is an optional pasting, thanks to \cref{eq:conestabcond}: on $F$, we have $\condE{\alpha Y^0}{\cF_t} = \condE{Z^0}{\cF_t}$, and so the factor in parentheses has conditional $\cF_t$-expectation of 1 on $F$. We shall apply (i) to deduce that  $\wh \bQ \in \cQ$, and for this we must show that 
\[ \bE_{\wh \bQ} [\V|\cF_\tau] = \bE_{ \bQ^Z} [\V|\cF_\tau].\]
We compute the left hand side to be
\begin{align*}
\bE_{\wh \bQ} [\V|\cF_\tau] &= \frac{1}{\wh \Lambda_t} \condE{\wh \Lambda \V}{\cF_t} \one{F} + \V \one{F^c} 
\\&=  \condE{\left(\frac{\alpha \bE[Y^0]\Lambda^Y_{t+1} }{\bE[Z^0]\Lambda^Z_t} \right) \frac{\Lambda^Y}{\Lambda^Y_{t+1}} \V}{\cF_t} \one{F} + \V \one{F^c} 
\\&= \frac{1}{\condE{Z^0}{\cF_t}}\condE{\alpha  Y}{\cF_t} \one{F} + \V \one{F^c}
\end{align*}

Condition \eqref{eq:conestabcond} shows that, on $F$, $\condE{\alpha  Y}{\cF_t} = \condE{Z}{\cF_t}$, so we conclude that $\wh \bQ \in \cQ$. We repeat the above steps for stopping time $\sigma = T \one F +t \one{F^c}$, measures $\wh \bQ$ and $\bQ^W$,
\[ \wt \Lambda = \wh\Lambda \one{F} + \wh \Lambda_t \left(\frac{\beta \bE[W^0]\Lambda^W_{t+1} }{\bE[Z^0]\Lambda^Z_t} \right) \frac{\Lambda^W}{\Lambda^W_{t+1}} \one{F^c} . \]
Condition \eqref{eq:conestabcond} gives that $\bE_{\wt \bQ} [\V|\cF_\tau] = \bE_{ \wh\bQ} [\V|\cF_\tau]$, and so $\wt \bQ \in \cQ$ by (i). It is simple to show that $X = \bE[Z^0] \wt \Lambda \V$, and thus $X \in \cD$ as required.

\emph{(ii)$\implies$(i):} Say (ii) holds; then (i) holds  for when $\tau = T$ trivially. Now suppose that (i) holds for any stopping time $\tau \ge k+1$ a.s., and proceed by backward induction on the lower bound of the stopping times. Fix an arbitrary stopping time $\wt \tau \ge k$ a.s., and define $F=\{ \wt \tau \ge k+1\}$ and the stopping time $\tau^* := \wt \tau \one F + T \one{F^c}$. Note that $\tau^* \ge k+1$, since $F^c = \{ \wt \tau = k\}$.

We shall now take $\bQ^1, \bQ^2 \in \cQ$ and $\wt \bQ \in \bQ^1 \oplus^{\text{opt}}_{\wt \tau} \bQ^2$ that satisfy  \cref{eq:pmstabcond}, and aim to show that $\wt \bQ$ is indeed an element of $\cQ $, with the help of condition (ii).
Define $\Lambda^i = d\bQ^i / d\P$ for $i=1,2$. Take a pasting of $\bQ^1$ and $\bQ^2$ at time $\tau^*$, $ \bQ^* \in \bQ^1 \oplus^{\text{opt}}_{ \tau^*} \bQ^2$, with Radon-Nikodym derivative
\[ \Lambda^* = \Lambda^1_{\tau^*} R^* \frac{\Lambda^2}{\Lambda^2_{(\tau^*+1)\wedge T}} \]
with $R^* \in \cL^1_+(\cF_{(\tau^*+1)\wedge T})$ and $\condE{R^*}{\cF_{\tau^*}}=1$. 
We note that  $\wt \Lambda := d \wt \bQ / d\P$ can be written as
\[ \wt \Lambda =  \Lambda^1_{\wt \tau} \wt R \frac{\Lambda^2}{\Lambda^2_{(\wt \tau+1)\wedge T}} = \Lambda^1_{\tau^*} \wt R \frac{\Lambda^2}{\Lambda^2_{(\tau^*+1)\wedge T}} \one F + \Lambda^1 \one{F^c}. \]
Set $X = \wt \Lambda \V$, $W=Z=\Lambda^1\V$,  $Y=\Lambda^2$, $\alpha=\wt R / R^*$, $\beta=1$ to satisfy the hypothesis of (ii). Thus, $X \in \cD$, whence $\wt \bQ \in \cQ$. This completes the inductive step.

\end{proof}

\begin{proof}[\textbf{Proof of \Cref{lem:A1}}]
The inclusion $\cD  \subset \cap_{t=0}^{T} \cM_t(\cD )$ is trivial. In the following, we write $Z|_t$ for $\condE{Z}{\cF_t}$.

Now $Z \in \cap_{t=0}^{T} \cM_t(\cD )$, and we aim to show that $Z\in \cD $. So, for all $t \in \{0,1,\dots, T\}$, there exist $\beta_t \in  \ul{L}^0_{t,+}$ and $Z^t \in \cD $ such that $\beta_t Z \in \ul{\cL}^1 \text{ and } Z|_{t}= \beta_t Z^t|_{t}$.

Define
\begin{align*}
\xi^{T}&=Z^{T} \\
\xi^t &= \one{F_t} \kappa_t \xi^{t+1}+\one{F^c_t}Z^t &\for t \in \{0,1,\dots, T-1\},
\end{align*}
where $F_t = \{ \beta_t>0\} $ and $\kappa_t = \beta_{t+1}/\beta_t$.

Note $ Z=Z|_T=\beta_{T}Z^{T}|_T=\beta_{T}\xi^{T}$ and 
\[Z=\beta_0 \kappa_0 \kappa_1 \cdots \kappa_{T-2} \xi^{T-1} = \beta_0 \xi^0. \]
Thus we only need to show $\xi^0$ is in the cone $\cD $ to deduce that $Z=\beta_0 \xi^0$ is in $\cD $. 
\paragraph{Claim} For all $t \in \{0,1,\dots, T\}$, we have $\xi^t|_{t}=Z^t|_{t}$ and $\xi^t \in \cD $.

We shall proceed by backwards induction, starting from the observation $\xi^{T} = Z^{T} \in \cD $. Suppose that for $s \ge t+1$, we have $\xi^s|_{s}=Z^s|_{s}$ and $\xi^s \in \cD $.
\begin{align*}
\xi^t|_{t} &= \condE{\one{F_t} \kappa_t \xi^{t+1}+\one{F^c_t}Z^t}{\ul{\cF}_{t}}
\\&=\condE{\one{F_t} \kappa_t Z^{t+1}+\one{F^c_t}Z^t}{\ul{\cF}_{t}}
\end{align*}
Now, if $\beta_t>0$, i.e. on the event $F_t$,
\[ \condE{ \kappa_t Z^{t+1}}{\ul{\cF}_{t}} = \frac{1}{\beta_t}\condE{ \beta_{t+1} Z^{t+1}}{\ul{\cF}_{t}} = \frac{1}{\beta_t}\condE{ Z|_{t+2}}{\ul{\cF}_{t}} = \frac{Z|_{t+1}}{\beta_t}  = Z^{t}|_{t} \]
allowing us to conclude  
\[ \xi^t|_{t} = \condE{\one{F_t} Z^{t}|_{t}+\one{F^c_t}Z^t}{\ul{\cF}_{t}}= Z^{t}|_{t}.\]
By hypothesis $\cD $ is stable, so by the alternative definition of stability (\Cref{prop:eqstab}), we see that $\xi^t \in \cD $.
\end{proof}

\begin{proof}[\textbf{Proof of \Cref{lem:A2}}]
It is clear that $[\cD ]$ is a closed convex cone in $\ul{\cL}^1$. To see that $[\cD ]$ is stable, we use  \Cref{prop:eqstab}. Fix $t\in \{0,1,\dots, T\}$, and suppose $Y,W \in [\cD ]$ are such that there exists $Z\in[\cD ]$, a set $F\in \ul{\cF}_t$, positive random variables $\alpha, \beta \in \ul{\cL}^0(\ul{\cF}_t)$ with $\alpha Y, \beta W \in \ul{L}^1(\R^{d+1})$ and 
\[ X:= \alpha Y  \one{F} +  \beta W\one{F^c}\]
satisfies $\condE{X}{\ul{\cF}_t}=\condE{Z}{\ul{\cF}_t}$. We aim to show $X$  is also a member of $[\cD ]$, that is,
\[X \in   \ol{\conv} \cM_s(\cD ) \qquad \forall  0 \le s  \le T-1 .\]

First consider $s \in \{0,1,\dots,  t-1 \}$. From the definition of $\cM_s(\cD )$,
\[ Z \in \conv\cM_s(\cD ) \q{and} \condE{X}{\ul{\cF}_t} = \condE{Z}{\ul{\cF}_t}  \qiq X \in\conv \cM_s(\cD ), \]
since the membership of an integrable $Z$ in $\cM_s(\cD )$ only depends on its conditional expectation $ \condE{Z}{\ul{\cF}_{s}}$.
More generally, we show
\[ Z \in  \ol{\conv}\cM_s(\cD ) \q{and} \condE{X}{\ul{\cF}_t} = \condE{Z}{\ul{\cF}_t}  \qiq X \in  \ol{\conv}\cM_s(\cD ). \]
 Take a sequence $(Z^n) \subset \conv \cM_s(\cD )$ such that $Z^n \to Z$ in $\ul{\cL}^1$. Define the sequence
\[ X^n :=  \condE{Z^n}{\ul{\cF}_t} + X-  \condE{X}{\ul{\cF}_t} .\]
Note that $X^n \to X$ as $n \to \infty$ and for each $n$, $\condE{X^n}{\ul{\cF}_t} = \condE{Z^n}{\ul{\cF}_t}$. So $X^n \in \conv \cM_s(\cD )$, thus $X \in  \ol{\conv}\cM_s(\cD )$.

Now consider $s \in \{t,t+1,\dots,  T-1 \}$. We begin by choosing sequences $(Y^n), (W^n) \subset \conv \cM_s(\cD )$ such that $Y^n \to Y$ and $W^n \to W$ in $\ul{\cL}^1$. Define, for $n,K \in \bN$,
\[ X^{n,K}:= \ind{\alpha \le K}  \alpha Y^n\one{F} + \ind{\beta \le K}  \beta W^n\one{F^c}. \]
The fact that $X^{n,K} \in \conv \cM_s(\cD )$ follows from the following two elementary properties:
\begin{enumerate}
\item
if $Z \in \conv\cM_s(\cD )$ and $g \in \ul{\cL}^\infty_+(\ul{\cF}_t)$, then $gZ \in \conv\cM_s(\cD )$;\footnote{Let $Z \in \cM_s(\cD )$. Then 
\begin{align*}
\exists \alpha_t \in \ul{L}^0_{t,+},  \exists Z'\in \cD  &\text{ such that } \alpha_t Z \in \ul{\cL}^1 \text{ and } Z|_{t}= \alpha_t Z'|_{t}
\\\implies\exists \alpha_tg \in \ul{L}^0_{t,+},  \exists Z'\in \cD  &\text{ such that } \alpha_t gZ \in \ul{\cL}^1 \text{ and } gZ|_{t}= \alpha_tg Z'|_{t}
\end{align*}and then take convex hulls.} and
\item
if $Z^i \in \conv\cM_s(\cD )$ for $i=1,2$, then $Z^1+Z^2 \in \conv\cM_s(\cD )$.
\end{enumerate}

Now, for any $K$ fixed, $\ind{\alpha \le K}  \alpha Y^n \to \ind{\alpha \le K}  \alpha Y$ as $n \to \infty$, and similarly $ \ind{\beta \le K}  \beta W^n \to  \ind{\beta \le K}  \beta W$. Since $\alpha Y$ and $\beta W$ are integrable, we now send $K \to \infty$ to see that
\[ X = \lim_{K\to \infty} \lim_{n \to \infty} X^{n,K} \in  \ol{\conv} \cM_s(\cD ) \]
which completes the proof that $X$ is indeed a member of $[\cD ]$.

To show minimality of $[\cD ]$ in the class of stable closed convex cones containing $\cD $, we note that if $\cD  \subset \cD '$ then $ [\cD ] \subset [\cD ']$. Taking $\cD '$ to be another stable closed convex cone containing $\cD $, we have $\cD '=[\cD ']$ by \Cref{lem:A1}, and so $\cD '$ contains $[\cD ]$. 
To show the equivalence in statement (b), the forward implication is due to the stability of $[\cD ]$, and the reverse is \Cref{lem:A1}.
\end{proof}

\begin{proof}[Proof of \Cref{thm:crucialClaim}] We set $\cB =\cA(\mathbf{V})$,   as above. 

First we prove that $\cM_t(\cB^*) \subset K_t(\cB)^*$. For arbitrary $ Z\in \cM_t(\cB^*)$, there exist $Z'\in \cB^*$ and $\alpha \in\ul L^0_+(\cF_t)$ with $\alpha Z' \in\ul \cL^1$ and $Z|_{t}=\alpha Z'|_{t}$. 

Note that, for any $X \in K_t(\cB)$,
\[ \bE[Z\cdot X]=\bE[Z|_{t}\cdot X] =\bE[\alpha Z'|_{t}\cdot X] = \lim_{n \to \infty} \bE[(\alpha \ind{\alpha \le n}X)\cdot Z'|_{t}] \le 0,\]
since $\alpha \ind{\alpha \le n} X\in \cB$ and $Z'\in \cB^*$. Hence $Z \in K_t(\cB)$, and since $Z$ is arbitrary, we have shown that $\cM_t(\cB^*) \subset K_t(\cB)^*$. 

For the reverse inclusion, $\cM_t(\cB^*)^* \subset K_t(\cB)$, note that $\cB^* \subset \cM_t(\cB^*) $ implies $ \cM_t(\cB^*)^*  \subset \cB$, and 
\begin{align*}
\ul{\cL}^\infty_+(\ul{\cF}_t) \cM_s(\cD ) =\cM_s(\cD ) &\qiq  \for X\in \cM_t(\cB^*)^*, \quad  g \in \ul{\cL}^\infty_+(\ul{\cF}_t),\quad \bE[X\cdot gZ] \le 0
\\& \qiq \ul{\cL}^\infty_+(\ul{\cF}_t) \cM_t(\cB^*)^* = \cM_t(\cB^*)^*.
\end{align*}
Define \[\cB_t:= \{ X \in\ul \cL^\infty(\ul\cF_T,\R^{d+1}): g X \in \cB \text{ for any } g \in\ul L^\infty_+(\ul\cF_t) \}.\]
Thus $\cM_t(\cB^*)^* \subseteq \cB_t$.
To finish the proof, we need only show that $X \in \cM_t(\cB^*)^*$ is $\cF_{t}$-measurable, since $\cB_t \cap\ul \cL^\infty(\ul\cF_{t},\R^{d+1}) = K_t(\cB)$.

To this end, note that for any $Z \in\ul \cL^1(\R^{d+1})$, it is true that $Z-Z|_{t} \in \cM_t(\cB^*)$, whence $\bE[(Z-Z|_{t})\cdot X] \le 0$. We deduce that 
\[ \bE[(Z-Z|_{t})\cdot X] = \bE[(X-X|_{t})\cdot Z] \le 0 \qquad \forall Z \in \ul\cL^1(\R^{d+1}),\]
and $X= X|_{t}$ $\P$-a.s.
\end{proof}


\end{document}